\documentclass[pra,onecolumn,notitlepage,nofootinbib,floatfix]{revtex4-2}

\usepackage{fullpage}
\usepackage{caption}
\usepackage[T1]{fontenc}
\usepackage[english]{babel}
\usepackage[utf8]{inputenc}
\usepackage[dvipsnames]{xcolor}
\usepackage[colorlinks=true,urlcolor=blue,citecolor=blue,linkcolor=blue,anchorcolor=blue]{hyperref}
\usepackage{graphics, graphicx, url, color, physics, cancel, multirow, bbm, soul, mathtools, mathrsfs, amsfonts, amssymb, amsmath, amsthm, dsfont, tabularx, bm, verbatim, theoremref, tikz, caption, subcaption, xcolor}
\usepackage[normalem]{ulem}
\usepackage[ruled,vlined]{algorithm2e}

\newtheorem{theorem}{Theorem}

\newtheorem{lemma}{Lemma}

\def\GL{{\textrm{GL}}}
\def\P{{\mathbf{P}}}
\def\Pl{{\mathcal{P}_\lambda}}
\def\Q{{\mathbf{Q}}}
\def\Qld{{\mathcal{Q}^d_\lambda}}
\def\Ql{{\mathcal{Q}}}
\def\Qlt{{\mathcal{Q}^2_\lambda}}
\def\Ud{{\mathcal{U}_d}}
\def\Sn{{\mathcal{S}_n}}

\def\CGl{{U^\textrm{CG}_\lambda}}
\def\SCG{{U^\textrm{S-CG}}}
\def\Sch{{U^\text{Sch}}}
\def\Cdn{{(\mathbb{C}^d)^{\otimes n}}}

\begin{document}

\title{Weak Schur sampling with logarithmic quantum memory}

\author{Enrique Cervero-Mart\'{i}n}
\email{enrique.cervero@u.nus.edu}
\affiliation{Centre for Quantum Technologies, National University of Singapore}

\author{Laura Mančinska}
\email{mancinska@math.ku.dk}
\affiliation{Centre for the Mathematics of Quantum Theory, University of Copenhagen}

\date{\today}

\begin{abstract}
The quantum Schur transform maps the computational basis of a system of $n$ qudits onto a \textit{Schur basis}, which spans the minimal invariant subspaces of the representations of the unitary and the symmetric groups acting on the state space of $n$ $d$-level systems. 
We introduce a new algorithm for the task of weak Schur sampling. Our algorithm efficiently determines both the Young label which indexes the irreducible representations and the multiplicity label of the symmetric group. 
There are two major advantages of our algorithm for weak Schur sampling  when compared to existing approaches which proceed via quantum Schur transform algorithm or Generalized Phase Estimation algorithm. First, our algorihtm is suitable for streaming applications and second it is exponentially more efficient in its memory usage.  
We show that an instance of our weak Schur sampling algorithm on $n$ qubits to accuracy $\epsilon$ requires only $O(\log_2n)$ qubits of memory and $O(n^3\log_2(\frac{n}{\epsilon}))$ gates from the Clifford+T set.
Further, we show that our weak Schur sampling algorithm on $n$ qudits decomposes into $O\big(dn^{2d}\log_2^p\big(\frac{n^{2d}}{\epsilon}\big)\big)$ gates from an arbitrary fault-tolerant qudit universal set, for $p\approx 4$, and requires a memory of $O(\log_dn)$ qudits to implement.
\end{abstract}

\maketitle

\section{Introduction}

Symmetry plays a very important role in quantum mechanics which quantum physicists and quantum computer scientists alike often strive to exploit in order to understand and manipulate quantum systems.
At the core of these symmetries are the groups of unitary matrices and of permutations.
The actions of the unitary group and the symmetric group on the state space of $n$ $d$-dimensional systems, or {qudits}, can be understood by {Schur-Weyl duality}.

Schur-Weyl duality allows us to to decompose the $n$-qudit state space in terms of the irreducible representations of the unitary group and the symmetric groups. 
Explicitly, Schur-Weyl duality~\cite{GoWa98} states that
\begin{align}\label{SWDualityIntro}
    \Cdn \overset{\Sn\times\Ud}{\cong} \bigoplus_{\lambda} \Pl\otimes\Qld,
\end{align}
where $\Pl$ and $\Qld$ are the irreducible representations of the symmetric and unitary groups, respectively, and the summation index $\lambda$ is known as the {Young label} and corresponds to partitions of $n$.
Any basis that spans $\Cdn$ and respects the decomposition in the right hand side of Eq.~\eqref{SWDualityIntro} is known as a \emph{Schur basis}.
Henceforth, we will refer to the Schur basis introduced by Bacon, Chuang and Harrow in \cite{BCH05, HarrowTh05} as \emph{the} Schur basis.
An element in the Schur basis is specified by a triple of labels $(\lambda, p_\lambda, q_\lambda)$ which respectively index the Young label $\lambda$ and the elements $p_\lambda$ and $q_\lambda$ within the bases of the irreducible representations $\Pl$ and $\Qld$.
A \emph{quantum Schur transform} maps the computational basis of $\Cdn$ to a Schur basis.
Likewise, we will use \emph{the} quantum Schur transform to denote the map introduced by the same authors in \cite{BCH05, HarrowTh05}.

Quantum Schur transforms have found a lot of applications in quantum computation and information protocols that seek to exploit the symmetries of quantum mechanical systems. 
They have been used in spectrum estimation \cite{KeyWer01, ChMi06, DoWr15} and quantum state tomography \cite{DoWr16, Haetal17}, as the Young label of $n$ i.i.d. copies of an arbitrary qudit provides an asymptotically good estimate to its spectrum. 
Schur transforms have also found applications in quantum data compression in \cite{HaMa03, Hayashi16}, where the Young label of $n$ identically prepared qubits is used in the encoding function. 
In universal distortion-free entanglement concentration \cite{MaHa07, BlCrGo14}, local applications of a quantum Schur transform followed by measurements of the Young label are used to distill maximal entanglement from partially entangled states. 
In addition, Schur transform and Schur-Weyl duality are also used in protocols for encoding and decoding into decoherence-free subspaces \cite{ZaRa97, KnLaVi00, BaconThesis03, KeBaLiWha01}, quantum superreplication \cite{ChiYa15, ChiYa16} and in de Finetti theorems \cite{KoeMit09, Gross21}.
Most recently, Schur transform have been used as a subroutine in geometric quantum machine learning \cite{Ragoneetal22, Nguyenetal22, Schatzkietal22} and quantum majority voting \cite{Buhrmanetal22}.

Most of the protocols outlined above only require knowledge of the Young label $\lambda$ of the Schur basis, and not the full triple $(\lambda, p_\lambda, q_\lambda)$ identifying a Schur basis element.
The task of measuring the Young label of a system of $n$ qudits is known as weak Schur sampling. 
Conversely, strong Schur sampling denotes the task of measuring the entire Schur basis labels $(\lambda, p_\lambda, q_\lambda)$.

One way to implement Schur sampling is to first apply the quantum Schur transform and then apply an appropriate coarse-graining of a standard basis measurement.

In the seminal work \cite{BCH05, HarrowTh05}, Bacon, Chuang and Harrow construct efficient quantum circuits to implement the Clebsch-Gordan transforms which they use to iteratively build the quantum Schur transform.
Their algorithm implements the quantum Schur transform to accuracy $\epsilon$ in time $O(\textrm{poly}(n,d,\log(1/\epsilon)))$ and uses $O(n)$ qubits of memory.
More recently, Kirby and Strauch proposed in \cite{KiSt18} a fully explicit implementation of the quantum Schur transform which is also based on the Clebsch-Gordan transforms.
Their implementation runs in $O(n^4\log(n/\epsilon))$ for qubits and $O(d^{1+p}n^{3d}\log^p(dn/\epsilon))$ for $p\approx 4$ for qudits, and requires $n+\log_d n$ qudits of memory.
Concurrently, in \cite{Krovi19} Krovi devised an implementation of another quantum Schur transform for qubits using quantum Fourier transforms as the building blocks. 
This implementation runs in $O(\textrm{poly}(n, \log(1/\epsilon)))$ and also requires $O(n)$ qubits of memory.

Another known approach for performing weak Schur sampling is via the generalized phase estimation (GPE) algorithm by Harrow \cite{HarrowTh05}. As compared to the approach via the quantum Schur transform, the approach via GPE achieves an improved runtime of $\text{poly}(n)+O(n\log d)$ but still requires $O(n)$ qubits of memory.

In this article, we present an algorithm for weak Schur sampling which, like the GPE algorithm, directly outputs the Young label of a system of $n$ qudits and may be modified to additionally output the register of the Schur basis which indexes the irreducible representations of the symmetric group.
Our implementation performs weak Schur sampling to accuracy $\epsilon$ on $n$ qubits using $O(n^3\log_2(\frac{n}{\epsilon}))$ gates from the Clifford+T set.
We further show that for an arbitrary fault tolerant qudit universal set, our weak Schur sampling algorithm for $n$ qudits decomposes into $O(n^{2d}\log_2^p(\frac{n}{\epsilon}))$ gates, for $p\approx 4$.
Crucially our algorithm requires only $\log_d n$ qudits to be implemented, which constitutes an exponential memory saving over the two existing approaches to performing weak Schur sampling (see Table~\ref{table}). Moreover, our weak Schur sampling protocol is also \emph{streaming} (or \emph{online}), that is, it may be performed adaptively as qudits progressively become available to the quantum device, with the possibility of being terminated at any iteration step while still outputting the Young label of the ensemble at termination. 

\begin{table}[h!]
\caption{Requirements for various algorithms for weak Schur sampling}
\label{table}
\centering
 \begin{tabular}{||c c c c||} 
 \hline
  & GPE \cite{HarrowTh05} &  Quantum Schur transform \cite{KiSt18}  &  \phantom{m}  Our algorithm\\ [0.5ex] 
 \hline\hline
 Gate count & $O(\text{poly}(n))$ & $O(n^4\log(\frac{n}{\epsilon}))$ & $O(n^3\log(\frac{n}{\epsilon}))$\\ 
 Memory & $O(n)$ & $O(n)$ & $O(\log(n))$ \\ [1ex] 
 \hline
 \end{tabular}
\end{table}

We remark that the quantum algorithm for entanglement concentration presented in \cite{BlCrGo14} shares some of the key ideas which facilitate our log-space, streaming algorithm for weak Schur sampling. The focus of our work is weak Schur sampling. In contrast to our work, the authors of \cite{BlCrGo14} focus on entanglement concentration and do not present any explicit algorithms for weak Schur sampling or provide analysis of resource requirements or correctness for this task.

\textbf{Technical overview: }In \cite{HarrowTh05, BCH05, KiSt18}, the Schur basis for $n$ qudits is constructed sequentially by adding one qudit at every iteration and subsequently `raising' the Schur basis from $k$ to $k+1$ qudits, as follows:
The $k$-th iteration initializes with $k$ qudits in the Schur basis and one qudit in the standard basis. 
That is, the $k+1$ qudits lie in space $\Big(\bigoplus_\lambda \Pl\otimes\Qld\Big)\otimes\mathbb{C}^d$, where the summation runs over partitions $\lambda$ of $n$.
The $k$-th iteration then proceeds by applying a Clebsch-Gordan transform on each of the $\dim\Pl$ copies of space $\Qld\otimes\mathbb{C}^d$.
In essence, the Clebsch-Gordan transforms `raise' the Schur basis of the irrep $\Qld\subset(\mathbb{C}^d)^{\otimes k}$ tensor an additional qubit in $\mathbb{C}^d$ to the Schur basis of the collection of irreps $\bigoplus_{\lambda'}\mathcal{Q}_{\lambda'}^d\subset(\mathbb{C}^d)^{\otimes k+1}$, where $\lambda'$ are partitions of $k+1$ formed from $\lambda$.
Overall, the $k$-th iteration results in the Schur basis on $k+1$ qudits.
In our algorithm for weak Schur sampling, we rely on the observation that performing a measurement collapsing 
$\bigoplus_{\lambda'}\mathcal{Q}_{\lambda'}^d \rightarrow \mathcal{Q}_{\lambda'}^d $
for some partition $\lambda'$ of $k+1$, reduces the size of the overall space that needs to be stored and saves on the number of Clebsch-Gordan transforms that need to be applied in each iteration. 
The main technical difficulty is thus showing that the probability of weak Schur sampling a specific Young label $\lambda$ partitioning $n$, after an application of the Schur transform followed by an appropriate measurement equals the probability of obtaining the same label $\lambda$ with our sequential algorithm, which during every iteration performs a single Clebsch-Gordan transform followed by a measurement.

This article is structured as follows: In Section \ref{Sec:Rev} we provide necessary background on Schur-Weyl duality, the quantum Schur transform and Schur sampling, and in Section \ref{Sec:wSt} we describe our algorithm, prove its correctness and give a fully explicit implementation for qudits.

\section{Review of the quantum Schur transform and Schur sampling}
\label{Sec:Rev}

We assume familiarity with representation theory including irreducible representations, isotypic decompositions and the representations of the symmetric and unitary groups, and with the framework of quantum information including quantum states, completely-positive trace-preserving (CPTP) maps and measurements. 
We refer to \cite{GoWa98} for a comprehensive introduction to representation theory, and \cite[Chapter\ 5]{HarrowTh05} for the representation theory specific to the quantum Schur transform.
We further refer to \cite{NiCh10} for an introduction to quantum information and computation. 

\subsection{Schur-Weyl Duality}
We consider the state space of $n$ qudits, i.e. the Hilbert space $\Cdn$. This space is spanned by the orthonormal basis $\{\ket{i_1}\otimes\cdots\otimes \ket{i_n}\}$ where each $i_k\in\{0,\ldots, d-1\}$.
We endow $\Cdn$ with a representation $\P$, of the symmetric group $\Sn$ of permutations of $n$ elements, and a representation $\Q$, of the unitary group $\Ud$ of unitary endomorphisms of $\mathbb{C}^d$.
The representation $\P: \Sn \rightarrow \textrm{GL}(\Cdn)$ is given by
\begin{align}\label{SymRep}
    \P(\sigma) \ket{i_1}\otimes \cdots \otimes\ket{i_n} = \ket{i_{\sigma^{-1}(1)}}\otimes \cdots \otimes \ket{i_{\sigma^{-1}(n)}},
\end{align}
while the representation $\Q:\Ud\rightarrow \GL(\Cdn)$ is given by
\begin{align}\label{UnRep}
    \Q(U) \ket{i_1}\otimes \cdots \otimes\ket{i_n} = U\ket{i_1}\otimes \cdots \otimes U\ket{i_n}.
\end{align}

The representations $\P$ and $\Q$ commute and moreover generate algebras $\mathscr{Q}:= \textrm{span}\{\Q(U) : U\in \Ud\}$ and $\mathscr{P}:= \textrm{span}\{\P(\sigma) : \sigma \in \Sn\}$ that are each others commutants\footnote{
The commutant of an algebra $\mathscr{A}$ in $\Cdn$ is defined as $\textrm{Comm}(\mathscr{A}):=\{M \in \textrm{End}(\Cdn) : MA = AM, \: \forall A\in\mathscr{A}\}$.
}.
Schur-Weyl duality states that $\textrm{Comm}(\mathscr{Q})=\mathscr{P}$ and $\textrm{Comm}\mathscr{P}=\mathscr{Q}$.
As a consequence we arrive at the well known formulation of Schur-Weyl duality:
\begin{align}\label{SWDuality}
    \Cdn \overset{\Sn\times\Ud}{\cong} \bigoplus_{\lambda\vdash n} \Pl\otimes\Qld,
\end{align}
where $\Pl$ and $\Qld$ are the respective irreducible representations (irreps) of the symmetric and unitary groups, and the summation index $\lambda$ runs through partitions of $n$ with exactly $d$ elements\footnote{A partition of $n$ is a tuple of non-increasing, non-negative integers that add up to $n$.} (written $\lambda\vdash n$) and are known as the \emph{Young labels}.
For each partition $\lambda$, the spaces $\Pl\otimes\Qld$ are called $\lambda$-isotypic subspaces.

Expressed as in Eq. \eqref{SWDuality}, Schur-Weyl duality states that the multiplicity spaces of the irreps of the symmetric group are isomorphic to the irreps of unitary group, and vice-versa.
That is, each $\lambda$-isotypic subspace $\Pl\otimes\Qld$ comprises of $\dim \Pl$ copies of irrep $\Qld$ or equivalently, $\dim \Qld$ copies of irrep $\Pl$, where
\begin{align}
    \dim\Pl=\frac{n!}{\prod_i h_\lambda(i)},
\end{align} 
is the hook length formula and
\begin{align}
        \dim\Qld = \frac{\prod_{0\leq i<j\leq d-1}(\lambda_i-\lambda_j+i-j)}{\prod_{m=1}^{d} m!}.
    \end{align}
is Stanley's hook-content formula \cite{Stanley71, BCH05}.
In the qubit case we have $\lambda = (\lambda_0, \lambda_1)$ and the above simplify to
\begin{align}
    \dim\Pl ={\lambda_0 + \lambda_1 \choose \lambda_0}\cdot \frac{\lambda_0 - \lambda_1 + 1}{\lambda_0+1}, \qquad\qquad \dim \Qld = \lambda_0 - \lambda_1 + 1.
\end{align}

\subsection{The quantum Schur transform}

A \emph{quantum Schur transform} $U_\textrm{Sch}(n)$ is an unitary map performing the isomorphism in Eq. \eqref{SWDuality}.
It maps the standard computational basis of $\Cdn$ to a \emph{Schur basis} $\{\ket{(\lambda, p_\lambda, q_\lambda)}\}_{\lambda, p_\lambda, q_\lambda}$, where the index $\lambda$ labels the isotypic subspace $\Pl\otimes\Qld$ and, with a slight abuse of notation\footnote{$\Pl$ and $\Qld$ are linear subspaces corresponding to some irreducible representations of respective dimensions $\dim\Pl, \dim\Qld$, as such we should write $p_\lambda\in[\dim\Pl]$ and $q_\lambda\in[\dim\Qld]$. }, the indices $p_\lambda\in\Pl$ and $q_\lambda\in\Qld$ label the respective bases elements of $\Pl$ and $\Qld$.
A Schur basis element $\ket{(\lambda, p_\lambda, q_\lambda)}$ may be expressed as a linear combination of standard computational basis elements
\begin{align}\label{Schur_basis}
    \ket{(\lambda, p_\lambda, q_\lambda)} = \sum_{i_1,\ldots,i_n=0}^{d-1} \big[U_\textrm{Sch}(n)\big]^{\lambda, p_\lambda, q_\lambda}_{i_1,\ldots,i_n}\ket{i_1,\ldots, i_n}.
\end{align}
In particular, in the computational basis $\Sch(n)$ is precisely the matrix with Schur basis vectors as rows.

There exist several implementations of quantum Schur transforms for qudits in quantum devices \cite{BCH05, HarrowTh05, KiSt18, Krovi19}. We give a brief description of the implementation in \cite{KiSt18} which uses the \emph{Clebsch-Gordan transforms} as building blocks.

Given $\lambda\vdash n$ and $\mu\vdash m$, the Clebsch-Gordan transform is the unitary map 
\begin{align}
    U^\textrm{CG}_{\lambda, \mu}: \mathcal{Q}_\lambda^d \otimes \mathcal{Q}_\mu^d \rightarrow \bigoplus_\nu \mathbb{C}^{n_\nu}\otimes \mathcal{Q}_\nu^d,
\end{align}
where $\nu\vdash (n+m)$, and $n_\nu$ is the multiplicity of irrep $\mathcal{Q}^d_\nu$ which is related to the so called \emph{Littlewood-Richardson coefficient} \cite{GoWa98}.
We are interested in the case when $m=1$ and $\nu$ runs over partitions of $n+1$. In this case
\begin{align}\label{CGTransform}
    \CGl: \mathcal{Q}_\lambda^d \otimes \mathcal{Q}_{(1)}^d \rightarrow \bigoplus_{j=0}^{d-1} \mathcal{Q}^d_{\lambda+\mathbf{e_j}},
\end{align}
where $\lambda+\mathbf{e_j}$ is the partition of $n+1$ obtained by adding a $1$ to the $j$-th element of the tuple $\lambda$ for $j\in\{0,\dots,d-1\}$, whenever this is valid\footnote{For the sake of exposition, we take $\mathcal{Q}^d_{\lambda+\mathbf{e}_j}$ to be empty if $\lambda+\mathbf{e}_j$ is not a valid partition} (e.g. $(1,1)+\mathbf{e_1}=(1,2)$ is not valid a partition of $3$). 
Since $\mathcal{Q}_{(1)}^d=\mathbb{C}^d$, Eq. \eqref{CGTransform} intuitively corresponds to `raising' the Schur basis of a single copy of irrep $\Qld$ upon the addition of an extra qudit, to the Schur basis of $n+1$ qudits in irreps $\bigoplus_{j=0}^{d-1} \mathcal{Q}^d_{\lambda+\mathbf{e_j}}$.

Parting from this observation, we may `raise' the Schur basis on $n$ qudits to that of $n+1$ qudits by applying a Clebsch-Gordan transform $U_\textrm{CG}^\lambda$ to each copy of each irrep $\Qld$ for $\lambda\vdash n$. 
By Eq. \eqref{SWDuality}, there are exactly $\dim \Pl$ copies of $\Qld$ for each $\lambda\vdash n$. 
Hence, defining the \emph{super Clebsch-Gordan transform} on $n$ qubits as
\begin{align}\label{SuperCG}
    \SCG(n) := \bigoplus_{\lambda\vdash n} \mathbb{I}_{\dim \Pl} \otimes \CGl,
\end{align}
it can be shown that (up to a reordering of the isotypic decomposition on each side)
\begin{align}
    \SCG(n): \Big(\bigoplus_{\lambda\vdash n} \Pl\otimes\Qld\Big) \otimes \mathbb{C}^d \rightarrow \bigoplus_{\lambda\vdash n+1} \Pl\otimes\Qld.
\end{align}
Informally, the super Clebsch-Gordan transform `raises' the Schur basis from $n$ qudits to $n+1$ qudits upon the addition of an extra qudit by virtue of applying a Clebsch-Gordan transform to each copy of each irrep $\Qld$ of $n$ qudits.

As such, the Schur transform based on Clebsch-Gordan transforms is built by sequential applications of super Clebsch-Gordan transforms
\begin{align}\label{SchPractical}
    \Sch(n)=\SCG({n-1})\cdot(\SCG({n-2})\otimes\mathbb{I}_d)\cdot...\cdot(\SCG({1})\otimes\mathbb{I}_d^{\otimes n-2}).
\end{align}

Constructing the quantum Schur transform in this manner yields a specific Schur basis known as the \textit{Gelfand-Tsetlin basis}.
This is the Schur basis considered in \cite{BCH05, HarrowTh05, KiSt18} and which we henceforth refer to as \emph{the} Schur basis when there is no ambiguity.
In the qubit case, the Gelfand-Tsetlin basis for the $\dim\Pl$ copies of $\Qlt$ for $\lambda = (\lambda_0,\lambda_1)\vdash n$ coincides with the $\dim\Pl$ bases of the $n$-qubit system with spin number $j=\frac{\lambda_0-\lambda_1}{2}$ and associated $2j+1=\lambda_0-\lambda_1+1=\dim\Qlt$ spin projection numbers $m_j\in\{-j,-j+1,\dots,j\}$.

\subsection{Schur sampling}
The task of measuring the partition $\lambda\vdash n$ in a system of $n$ qudits is known as \textit{weak Schur sampling}. To formalize this task, let us consider the isotypic spaces $\Pl\otimes\Qld$ indexed along with the projections
\begin{align}\label{wSs_std}
    \Pi_\lambda^\text{Std} = \sum_{p_\lambda\in \Pl}\sum_{q_\lambda\in \Qld} \ketbra{(\lambda, p_\lambda, q_\lambda)},
\end{align}
where the superscript \textit{Std} denotes that the projection is in the standard basis and $\ket{(\lambda, p_\lambda, q_\lambda)}$ are as in Eq. \eqref{Schur_basis}.
A quantum algorithm for weak Schur sampling is any procedure which upon all inputs $\rho$ outputs the partition label $\lambda$ with probability $\text{tr}\big[\rho \Pi_\lambda^\text{Std}\big]$.

One algorithm for weak Schur transform is to first apply the quantum circuit for the Schur transform to change to the Schur basis and then perform the projective measurement given by
\begin{align}\label{wSs}
    \Pi_\lambda^\text{Sch} 
    = \bigoplus_{\lambda'\vdash n}\delta_{\lambda\lambda'}\mathbb{I}_{\dim(\mathcal{P}_{\lambda'})}\otimes\mathbb{I}_{\dim(\mathcal{Q}^d_{\lambda'})},
\end{align} 
where the superscript \textit{Sch} denotes the projection is in the Schur basis. In the next section we present an exponentially more efficient algorithm in terms of its memory requirements.

We refer to $\Pi_\lambda^\text{Std}$ and $\Pi_\lambda^\text{Sch}$ as projections onto the `$\lambda$ register' of the Schur basis. 
Further we remark that $\Pi_\lambda^\text{Std}= (\Sch(n))^\dagger\cdot\Pi_\lambda^\text{Sch}\cdot \Sch(n)$ which confirms
\begin{align}\label{change_basis1}
    \text{tr}\big[\rho \Pi_\lambda^\text{Std}\big] 
    = \text{tr}\big[\rho \big((\Sch(n))^\dagger\cdot\Pi_\lambda^\text{Sch}\cdot \Sch(n)\big)\big]
    = \text{tr}\big[\tilde{\rho}\Pi_\lambda^\text{Sch}\big]
\end{align}
where $\rho$ and $\tilde{\rho}$ are the initial states in the standard and Schur basis respectively.

\section{An algorithm for weak Schur Sampling}\label{Sec:wSt}
In this section we introduce our algorithm for weak Schur sampling, a protocol which directly outputs the Young label without the need of a quantum Schur transform.
We will see that our algorithm is more efficient in terms of number of operations and significantly more efficient in the memory requirements than other methods which utilize quantum Schur transforms, or generalized phase estimation.
Thus it is an optimal alternative in any process requiring weak Schur sampling.

The following is a `high-level' implementation of our algorithm for weak Schur sampling:

\begin{algorithm}[H]\label{alg:streamalgo}
\SetAlgoLined
\caption{An algorithm for weak Schur sampling}
\SetKwInOut{Input}{Input}
\SetKwInOut{Output}{Output}
\SetKwInOut{Init}{Initialization}
\SetKwRepeat{Repeat}{repeat}{until}
\Input{A stream of $n$ qudits}
\Output{A partition $\lambda\vdash n$}
initialization\;
receive first qubit\;
$\lambda = (1)$\;
\For{$k=1$ to $n-1$}{
receive the $(k+1)$-st qudit\;
apply $\CGl: \Ql_\lambda^d\otimes\Ql^d_{(1)}\rightarrow\bigoplus_{j=0}^{d-1}\Ql^d_{\lambda+\mathbf{e}_j}$\;
apply measurement $\{\tilde{\Pi}_{\lambda+\mathbf{e}_j}^\text{Sch}\}_j$ to obtain $\Ql^d_{\lambda+\mathbf{e}_j}$ for a $j\in\{0,...,d-1\}$\;
$\lambda \leftarrow \lambda+\mathbf{e}_j$\;
}
\Return $\lambda$
\end{algorithm}

Recall that for $\lambda\vdash k$ the direct summation $\bigoplus_{j=0}^{d-1}\Ql^d_{\lambda+\mathbf{e}_j}$ in the algorithm above only ranges over valid partitions $\lambda+\mathbf{e}_j$ of $k+1$.
That is, we set $\Ql^d_{\lambda+\mathbf{e}_j} = \emptyset$ when ${\lambda+\mathbf{e}_j}$ is not a valid partition of $k+1$.

The crucial step in Algorithm \ref{alg:streamalgo} is the application of the projection $\{\tilde{\Pi}_{\lambda+\mathbf{e}_j}^\text{Sch}\}_j$, where for each $j\in\{0,...,d-1\}$, $\tilde{\Pi}_{\lambda+\mathbf{e}_j}^\text{Sch}$ is defined as:
\begin{align}\label{wwSs}
     \tilde{\Pi}_{\lambda+\mathbf{e}_j}^\text{Sch}:= \bigoplus_{\lambda'\vdash (k+1)}\delta_{\lambda+\mathbf{e}_j,\lambda'}\mathbb{I}_{\dim(\mathcal{Q}^d_{\lambda'})}.
\end{align}
The measurement $\{\tilde{\Pi}_{\lambda+\mathbf{e}_j}^\text{Sch}\}_j$ projects the direct sum $\bigoplus_{j=0}^{d-1}\Ql^d_{\lambda+\mathbf{e}_j}$ onto $\Ql^d_{\lambda+\mathbf{e}_j}$ for a particular value of $j$. 

This ensures that at the beginning of the $k$-th iteration of the algorithm, the $k$ qudit state will lie completely in (one specific copy of) the irrep $\Qld$ for known $\lambda \vdash k$. 
This implies a couple of things:
\begin{enumerate}
    \item Upon receipt of the $(k+1)$-th qudit, the algorithm performs the single Clebsch-Gordan transform $\CGl$. 
    \item The algorithm needs to store only $\Big\lceil \log_d\Big( \sum_{j=0}^{d-1} \dim(\Ql^d_{\lambda+\mathbf{e_j}})\Big)\Big\rceil$ qudits in each iteration.
\end{enumerate}
The first point implies that that Algorithm \ref{alg:streamalgo} is strictly less computationally intensive than a quantum Schur transform. The second point establishes the fact that our weak Schur sampling algorithm requires only a logarithmic number of qudits of memory, which constitutes an exponential improvement over quantum Schur transform and GPE algorithms. 
It remains to show that the probability of measuring the Young label $\lambda\vdash n$ after application of projections in Eq. \eqref{wSs_std} or Eq. \eqref{wSs} equals the probability of Algorithm \ref{alg:streamalgo} outputting the same label.

\begin{theorem}[Correctness of Algorithm \ref{alg:streamalgo}]
Let $\rho$ be an $n$ qudit state and let $\Lambda^\textrm{wSch}$ denote the output label of Algorithm \ref{alg:streamalgo}.
Then for all $\lambda\vdash n$ 
\begin{align}\thlabel{wSsCorrectness}
    \Pr[\Lambda^\textup{wSch} = \lambda] = \Tr[\rho \Pi_\lambda^\textup{Std}]
\end{align}
where $\Pi_\lambda^\textup{Std}$ are the projections in \eqref{wSs_std}.
\end{theorem}

\begin{proof}
Viewing the irreps $\Pl$ of the symmetric group as the multiplicity spaces of the irreps $\Qld$ of the unitary group, we can rewrite Schur-Weyl duality as
\begin{align}
    \Cdn \overset{\Sn \times \Ud}{\cong} \bigoplus_{\lambda\vdash n}\bigoplus_{p_\lambda\in\Pl} \mathcal{Q}^d_{\lambda, p_\lambda},
\end{align}
where the additional $p_\lambda$ label in $\mathcal{Q}^d_{\lambda, p_\lambda}$ is interpreted to index the copies of $\Qld$.

Fix Young label $\lambda\vdash n$ and recall the projection in the standard basis onto isotypic space $\Pl\otimes\Qld$ from Eq. \eqref{wSs_std}
\begin{align}
    \Pi_\lambda^\text{Std} = \sum_{p_\lambda\in \Pl}\sum_{q_\lambda\in \Qld} \ketbra{(\lambda, p_\lambda, q_\lambda)}.
\end{align}
We can similarly define 
\begin{align}
    \Pi^\text{Std}_{\lambda,p_\lambda}=\sum_{q_\lambda\in\Qld}\ketbra{(\lambda, p_\lambda, q_\lambda)}
\end{align}
as the projection in the standard basis onto the $p_\lambda$-th copy of $\Qld$.
Additionally let
\begin{align}
    \Pi^\text{Sch}_{\lambda}=\Sch(n)\cdot\Pi^\text{Std}_{\lambda}\cdot(\Sch(n))^\dagger, \qquad\text{and}\qquad    \Pi^\text{Sch}_{\lambda,p_\lambda}=\Sch(n)\cdot\Pi^\text{Std}_{\lambda,p_\lambda}\cdot(\Sch(n))^\dagger
\end{align}
be the respective projections in the Schur basis. 
Note that $\Pi^\text{Sch}_{\lambda}$ is also given in in Eq. \eqref{wSs}.

Now, define 
\begin{align}\label{isom_l_pl}
    U^\text{wSch}_{\lambda, p_\lambda}:\Cdn \rightarrow \mathcal{Q}^d_{\lambda, p_\lambda},\qquad U^\text{wSch}_{\lambda, p_\lambda} := \Pi^\text{Sch}_{\lambda,p_\lambda}\cdot\Sch(n)
\end{align}
to be a map from $\Cdn$ in the standard basis to $\mathcal{Q}^d_{\lambda, p_\lambda}$ in the Schur basis.
Similarly, let
\begin{align}\label{isom_pl}
    U^\text{wSch}_{\lambda}:\Cdn \rightarrow \mathcal{Q}^d_{\lambda},\qquad U^\text{wSch}_{\lambda} := \Pi^\text{Sch}_\lambda\cdot\Sch(n)
\end{align}
be a map from $\Cdn$ in the standard basis to the isotypic subspace $\Pl\otimes\Qld$ in the Schur basis.
Note that from these definitions we have both
\begin{align}\label{wSs_l_pl}
    \sum_{p_\lambda\in\Pl}\Pi^\text{Std}_{\lambda,p_\lambda} = \Pi^\text{Std}_{\lambda}, \qquad\qquad U^\text{wSch}_{\lambda} := \sum_{p_\lambda\in \Pl} U^\text{wSch}_{\lambda, p_\lambda}.
\end{align}

To be precise, we refrain from calling $U^\text{wSch}_{\lambda, p_\lambda}$ and $U^\text{wSch}_{\lambda}$ co-isometries\footnote{A linear map $A:V\rightarrow W$ between Hilbert spaces $V,W$ is a co-isometry if $AA^\dagger=\mathbb{I}_W$.} as we interpret the respective co-domains to be embedded into the full space $\bigoplus_\lambda \Pl\otimes\Qld$.
Intuitively, we view $U^\textrm{wSch}_\lambda$ as the $d^n\times d^n$ square matrix $\Sch(n)$ with zeros in all the rows which do not correspond to the basis of isotypic space $\Pl\otimes\Qld$, and $U^\text{wSch}_{\lambda, p_\lambda}$ as the $d^n\times d^n$ square matrix $\Sch(n)$ with zeros in all the rows which do not correspond to the basis of the $p_\lambda$-th copy of space $\Qld$.

From $U^\text{wSch}_{\lambda}$ we define the map\footnote{Note that this map is only CP and not CPTP as $U^\text{wSch}_{\lambda}$ are not co-isometries by our choice of co-domain. 
This choice is purely for ease of presentation as the same results would apply if $U^\text{wSch}_{\lambda}$ and $U^\text{wSch}_{\lambda, p_\lambda}$ were co-isometries, up to careful padding when multiplying maps together. } $\mathcal{U}^\text{wSch}_{\lambda}:\rho\mapsto U^\text{wSch}_{\lambda}\rho (U^\text{wSch}_{\lambda})^\dagger$ for states $\rho$ in the standard basis and note that 
\begin{align}\label{wSs_trace1}
    \text{tr}\big[\mathcal{U}^\text{wSch}_{\lambda}(\rho)\big] 
    &= \text{tr}\big[U^\text{wSch}_{\lambda}\rho(U^\text{wSch}_{\lambda})^\dagger\big] \\
    &= \text{tr}\big[\big(\Pi^\text{Sch}_\lambda\cdot\Sch(n)\big)\rho\big((\Sch(n))^\dagger\cdot\Pi^\text{Sch}_\lambda\big)\big] \\
    &= \text{tr}\big[{\rho}\Pi^\text{Std}_\lambda\big]\\
    &= \sum_{p_\lambda\in\Pl}\text{tr}\big[\rho\Pi^\text{Std}_{\lambda,p_\lambda}\big],
\end{align}
where in the penultimate equality we used the change of basis in Eq. \eqref{change_basis1} and in the last equality we used Eq. \eqref{wSs_l_pl}. 

Similarly, define the map $\mathcal{U}^\text{wSch}_{\lambda,p_\lambda}:\rho\mapsto U^\text{wSch}_{\lambda,p_\lambda}\rho (U^\text{wSch}_{\lambda,p_\lambda})^\dagger$ for states $\rho$ in the standard basis and note that
\begin{align}
    \text{tr}\big[\mathcal{U}^\text{wSch}_{\lambda,p_\lambda}(\rho)\big] 
    &= \text{tr}\big[U^\text{wSch}_{\lambda,p_\lambda}\rho(U^\text{wSch}_{\lambda,p_\lambda})^\dagger\big] \\
    &= \text{tr}\big[\big(\Pi^\text{Sch}_{\lambda,p_\lambda}\cdot\Sch(n)\big)\rho\big((\Sch(n))^\dagger\cdot\Pi^\text{Sch}_{\lambda,p_\lambda}\big)\big] \\
    &= \text{tr}\big[\rho\Pi^\text{Std}_{\lambda,p_\lambda}\big],
\end{align}
which readily implies 
\begin{align}
    \text{tr}\big[\mathcal{U}^\text{wSch}_{\lambda}(\rho)\big] = \sum_{p_\lambda\in\Pl}\text{tr}\big[\mathcal{U}^\text{wSch}_{\lambda,p_\lambda}(\rho)\big].\label{wSs_trace2}
\end{align}

Now, let $\pi(\lambda) = \{\lambda^1,...,\lambda^{n-1}, \lambda^n\}$ be an $n$-tuple of partitions such that
\begin{enumerate}
    \item $\lambda^k\vdash k$;
    \item $\lambda^{k+1} = \lambda^k + \mathbf{e}_j$ for some valid $j\in\{0,...,d-1\}$;
    \item $\lambda^n = \lambda$.
\end{enumerate}
We use $\pi(\lambda)$ to denote the `path' of Young labels traversed in Algorithm \ref{alg:streamalgo} from $\lambda^1=(1)$ to its output $\lambda\vdash n$.
Since Schur-Weyl duality (Eq. \eqref{SWDuality}) implies that each irrep $\Qld$ has $\dim \Pl$ copies, there are $\dim \Pl$ distinct paths\footnote{This also follows intuitively from the Young diagram representation of $\lambda=(\lambda_0,...,\lambda_{d-1})\vdash n$: a diagram of $d$ rows each containing $\lambda_i$ boxes which can be used to `represent' a particular irrep of $\Qld$. The number of paths is the number of ways to get to the Young diagram of $\lambda\vdash n$ from the Young diagram of $(1)$ by adding one box at a time in such a way that each resulting diagram corresponds to a valid partition. This corresponds to the dimension of $\Pl$ \cite{GoWa98} which is given by the hook's length formula.}
$\pi(\lambda)$, which we denote $\pi_{p_\lambda}(\lambda)$.

Fix partition $\lambda$ and index $p_\lambda$ and suppose an execution of Algorithm \ref{alg:streamalgo} traverses the path $\pi_{p_\lambda}(\lambda)$. 
Let $\rho^k$ be the $k$-qudit input state to the $k$-th iteration of the for loop in Algorithm \ref{alg:streamalgo}.
By construction, the state $\rho^k$ lives in the space $\mathcal{Q}^d_{\lambda^k}$ for $\lambda^k\in\pi_{p_\lambda}(\lambda)$.
The $k$-th iteration of Algorithm \ref{alg:streamalgo} corresponds to an application of the map
\begin{align}
    \tilde{\mathcal{U}}_{\lambda,p_\lambda}^k:\rho^k\mapsto \tilde{U}_{\lambda,p_\lambda}^k \rho^k (\tilde{U}_{\lambda,p_\lambda}^k)^\dagger,
\end{align}
where $\tilde{U}_{\lambda,p_\lambda}^k$ is given by 
\begin{align}
    \tilde{U}_{\lambda,p_\lambda}^k:\mathcal{Q}^d_{\lambda^k}\otimes\mathbb{C}^d\rightarrow\mathcal{Q}^d_{\lambda^{(k+1)}}, \qquad  \tilde{U}_{\lambda,p_\lambda}^k := \tilde{\Pi}_{\lambda^{k+1}}^\textrm{Sch}\cdot U^\textrm{CG}_{\lambda^k},
\end{align}
and $\tilde{\Pi}_{\lambda^{k+1}}^\textrm{Sch}$ is defined in \eqref{wwSs}.

Combining all $n-1$ iterations, Algorithm \ref{alg:streamalgo} corresponds to an application of the map
\begin{align}\label{wSt_cptp}
    \tilde{\mathcal{U}}_{\lambda,p_\lambda}^\text{wSch}:\rho\mapsto \tilde{U}_{\lambda,p_\lambda}^\text{wSch}\rho(\tilde{U}_{\lambda,p_\lambda}^\text{wSch})^\dagger,
\end{align}
where $\tilde{U}_{\lambda,p_\lambda}^\text{wSch}$ is defined by
\begin{align}\label{wst_isom}
    \tilde{U}_{\lambda,p_\lambda}^\text{wSch}:\Cdn\rightarrow\mathcal{Q}^d_{\lambda,p_\lambda}, \qquad \tilde{U}_{\lambda,p_\lambda}^\text{wSch}=\prod_{k=1}^{n-1} \Big(\tilde{\Pi}_{\lambda^{k+1}}^\textrm{Sch}\cdot U^\textrm{CG}_{\lambda^k}\Big) \otimes \mathbb{I}_d^{\otimes (n-k-1)},
\end{align}
where we interpret the product of non-commuting factors to be in the order defined in Algorithm \ref{alg:streamalgo}.
The identities are required for padding as iteration $k$ acts on $(k+1)$ out of the $n$ qudits, but moving forward they will be omitted.

To complete the proof, it suffices to show that $\tilde{U}_{\lambda,p_\lambda}^\text{wSch}=U_{\lambda,p_\lambda}^\text{wSch}$ as this would imply $\tilde{\mathcal{U}}_{\lambda,p_\lambda}^\text{wSch}=\mathcal{U}_{\lambda,p_\lambda}^\text{wSch}$.
Then, it would follow from equations \eqref{wSs_trace1}-\eqref{wSs_trace2} that
\begin{align}
    \text{tr}\big[{\rho}\Pi^\text{Std}_\lambda\big] 
    = \text{tr}\big[\mathcal{U}^\text{wSch}_{\lambda}(\rho)\big] 
    = \sum_{p_\lambda\in\Pl}\text{tr}\big[\mathcal{U}^\text{wSch}_{\lambda,p_\lambda}(\rho)\big]
    = \sum_{p_\lambda\in\Pl}\text{tr}\big[\tilde{\mathcal{U}}^\text{wSch}_{\lambda,p_\lambda}(\rho)\big].
\end{align}
The leftmost side of this expression is the probability that a weak Schur sampling protocol outputs partition $\lambda\vdash n$ on input $\rho$ in the standard basis.
The term $\text{tr}\big[\tilde{\mathcal{U}}^\text{wSch}_{\lambda,p_\lambda}(\rho)\big]$ on the right is precisely the probability that Algorithm \ref{alg:streamalgo} traverses path $\pi_{p_\lambda}(\lambda)$ and outputs partition $\lambda\vdash n$ on input $\rho$ in the standard basis and the summation over the indices $p_\lambda\in\Pl$ simply accounts for the possible choice of paths. 
That is, 
\begin{align}
    \sum_{p_\lambda\in\Pl}\text{tr}\big[\tilde{\mathcal{U}}^\text{wSch}_{\lambda,p_\lambda}(\rho)\big] = \Pr[\Lambda^\text{wSch}=\lambda].
\end{align}
In essence, we want to show that one execution of Algorithm \ref{alg:streamalgo} outputting partition $\lambda\vdash n$ and traversing $\pi_{p_\lambda}(\lambda)$ is equivalent to an application of (the map associated to) $U^\text{wSch}_{\lambda,p_\lambda}=\Pi^\text{Sch}_{\lambda,p_\lambda}\cdot\Sch(n)$, the $d^n\times d^n$ matrix $\Sch(n)$ with zeros everywhere except in the rows corresponding to the basis of $p_\lambda$-th copy of $\Qld$.
Hence, accounting for all paths $\pi_{p_\lambda}(\lambda)$ yields $\sum_{p_\lambda\in\Pl}U^\text{wSch}_{\lambda,p_\lambda}=\Pi^\text{Sch}_{\lambda}\cdot\Sch(n)$, the $d^n\times d^n$ matrix $\Sch(n)$ with zeros everywhere except in the rows corresponding to the basis of $\Pl\otimes\Qld$.

We show that $\tilde{U}_{\lambda,p_\lambda}^\text{wSch}=U_{\lambda,p_\lambda}^\text{wSch}$ by induction on the number of qudits:
For the base case of $\lambda\vdash 2$, it is clear that for the only possible path $\pi(\lambda) = \{(1), \lambda\}$
\begin{align}
    \tilde{U}_{\lambda,p_\lambda}^\text{wSch}
    &:= \tilde{\Pi}_\lambda^\text{Sch} U^\text{CG}_{(1)}
    = \Pi_{\lambda}^\text{Sch} \Sch(2)
    = {U}_{\lambda,p_\lambda}^\text{wSch},
\end{align}
which follows from the definition in Eq. \eqref{wst_isom} and the facts that $\Sch(2)=U^\text{CG}_{(1)}$ and $\Pi_{\lambda}^\text{Sch}=\tilde{\Pi}_\lambda^\text{Sch}$ for $\lambda\vdash 2$.

Now suppose that for $\lambda^n\vdash n$ and path $\pi_{p_{\lambda^n}}(\lambda^n)$ we have that $\tilde{U}_{{\lambda^n},p_{\lambda^n}}^\text{wSch}=U_{{\lambda^n},p_{\lambda^n}}^\text{wSch}$ and consider $\lambda = (\lambda^n + \mathbf{e_j}) \vdash (n+1)$ for some valid $j=0,...,d-1$ with corresponding path $\pi_{p{_\lambda}}(\lambda) = \pi_{p_{\lambda^n}}(\lambda^n)\cup\{\lambda\}$.
Then
\begin{align}
    \tilde{U}_{\lambda,p_\lambda}^\text{wSch}
    &= \prod_{k=1}^{n} \Big(\tilde{\Pi}_{\lambda^{k+1}}^\textrm{Sch}\cdot U^\textrm{CG}_{\lambda^k}\Big)\\
    &= \Big(\tilde{\Pi}^\text{wSch}_\lambda \cdot \CGl\Big) \cdot \Bigg(\prod_{k=1}^{n-1} \Big( \tilde{\Pi}_{\lambda^{k+1}}^\textrm{Sch}\cdot U^\textrm{CG}_{\lambda^k} \Big)\otimes \mathbb{I}_d\Bigg) \\
    &= \Big(\tilde{\Pi}^\text{wSch}_\lambda \cdot \CGl\Big) \cdot \Big( \tilde{U}_{\lambda^n,p_{\lambda^n}}^\text{wSch} \otimes \mathbb{I}_d \Big) \\
    &= \Big(\tilde{\Pi}^\text{wSch}_\lambda \cdot \CGl\Big)\cdot \Big( U_{\lambda^n,p_{\lambda^n}}^\text{wSch} \otimes \mathbb{I}_d \Big)
\end{align}
From the definition, $U_{\lambda^n,p_{\lambda^n}}^\text{wSch}:\Cdn\rightarrow \mathcal{Q}^d_{\lambda^n,p_{\lambda^n}}$ is the $d^n\times d^n$ matrix containing the $\dim\mathcal{Q}^d_{\lambda^n}$ vectors spanning $\mathcal{Q}^d_{\lambda^n,p_{\lambda^n}}$. 
Further, 
$\CGl:\mathcal{Q}^d_{\lambda^n,p_{\lambda^n}}\otimes\mathbb{C}^d\rightarrow \bigoplus_{j=0}^{d-1} \mathcal{Q}^d_{\lambda^n+\mathbf{e_j},p_{\lambda^n+\mathbf{e_j}}}$ maps the product basis of $\mathcal{Q}^d_{\lambda^n,p_{\lambda^n}}\otimes \mathbb{C}^d$ to the Schur basis of the direct sum $\bigoplus_{j=0}^{d-1} \mathcal{Q}^d_{\lambda^n+\mathbf{e_j},p_{\lambda^n+\mathbf{e_j}}}$.
Namely, the $d^{n+1}\times d^{n+1}$ matrix $\CGl\cdot U_{\lambda^n,p_{\lambda^n}}^\text{wSch}$ contains the basis vectors of the spaces $\mathcal{Q}^d_{\lambda^n+\mathbf{e_j}, p_{\lambda^n+\mathbf{e_j}}}$ for all $j\in\{0,...,d-1\}$ as rows.
Since $\tilde{\Pi}^\text{wSch}_\lambda$ picks out precisely the $\dim \Qld$ vectors corresponding to the basis of $\mathcal{Q}^d_{\lambda^n+\mathbf{e_j}, p_{\lambda^n+\mathbf{e_j}}}$ such that $\lambda = (\lambda^n+\mathbf{e_j})$, we conclude that $(\tilde{\Pi}^\text{wSch}_\lambda \cdot \CGl) U_{\lambda^n,p_{\lambda^n}}^\text{wSch} $ is the $d^{n+1}\times d^{n+1}$ matrix containing only the basis vectors of $\Qld$ and is therefore equivalent to $U_{\lambda,p_{\lambda}}^\text{wSch}$, as required.
\end{proof}

Note that in the proof of \thref{wSsCorrectness} we establish a one-to-one correspondence between the `path' $\pi(\lambda)$ of Young labels obtained through the $n-1$ iterations of Algorithm \ref{alg:streamalgo}, and each of the labels $p_\lambda\in\Pl$.
It follows that our algorithm may be modified to output the symmetric label $p_\lambda$ simply by looking at the resulting $\pi(\lambda)$.

In the remainder of the article, we propose an explicit implementation for qubits and analyse its memory requirements and runtime, which readily generalise to the qudit case.

\subsection{Weak Schur sampling algorithm for qubits}
During the $k$-th iteration of Algorithm \ref{alg:streamalgo} on qubits, the quantum device begins with irrep $\Qlt$ with $\lambda=(\lambda_0,\lambda_1)\vdash k$, applies $\CGl$ to obtain $\mathcal{Q}^2_{\lambda+\mathbf{e}_0} \oplus \mathcal{Q}^2_{\lambda+\mathbf{e}_1}$, and projects onto $\mathcal{Q}^2_{\lambda+\mathbf{e}_j}$ for $j\in\{0,1\}$.

To encode these operations, the quantum machine will maintain two registers, $L$ and $Q$ throughout each iteration. After application of the Clebsch-Gordan transform:
\begin{itemize}
    \item Register $L$ (for Lambda) encodes the Young label $\lambda+\mathbf{e}_j$ and requires a single qubit to be specified: $\ket{j}_L$ for $j\in\{0,1\}$.
    \item Register $Q$ (for $\Qld$) contains the state of the qubits within irrep $(\lambda+\mathbf{e}_j)\vdash(k+1)$. Register $Q$ needs to be at most $\lceil\log_2(k+2)\rceil$ qubits as for arbitrary $\lambda\vdash k$ we have $\dim \mathcal{Q}^2_\lambda = \lambda_0 - \lambda_1 +1 = k - 2\lambda_1+1$ which is maximized when $\lambda_1 = 0$.
\end{itemize}
Thus, iteration $k$ requires $\lceil\log_2(2k+4)\rceil$ qubits.

Following the Clebsch-Gordan transform, the quantum device determines $j$ by measuring the $L$ register and discards a qubit if 
\begin{align}\label{QubitRemoval}
    \lceil\log_2(2k+4)\rceil \neq \lceil\log_2(k+3)\rceil.
\end{align}
This is because the $(k+1)$-th iteration requires $\lceil\log_2(2(k+1)+4)\rceil$ qubits counting the additional $(k+2)$-th qubit it receives.
Thus, iteration $(k+1)$ must begin with $\lceil\log_2(2(k+1)+4)\rceil-1 = \lceil\log_2(k+3)\rceil$ qubits.
In particular, iteration $k$ begins with $\lceil\log_2(k+2)\rceil$ qubits describing the irrep $\Qlt$ for $\lambda\vdash k$.
Note that $\lceil\log_2(k+2)\rceil$ overestimates the number of qubits required to describe $\dim \mathcal{Q}^2_\lambda = k - 2\lambda_1+1\leq k+1$, so qubits need to be padded accordingly.
Since equality in \eqref{QubitRemoval} only occurs logarithmically often, we conclude that our algorithm requires only $\lceil\log_2(2(n-1)+4)\rceil=O(\log_2(n))$ qubits.

The explicit implementation of our weak Schur sampling algorithm on these registers is as in Algorithm \ref{alg:pracstreamalgo}.
In the algorithm, the binary tree \textbf{PartTree} (for partition tree) is maintained to determine the path $\pi(\lambda)$ for output $\lambda$. The path $\pi(\lambda)$ fully determines $p_\lambda\in \Pl$. 
The rearranging in steps $3$ and $6$ ensure that register $L$ encodes the information about partitions $\lambda+\mathbf{e}_0$ and $\lambda+\mathbf{e}_1$, as is illustrated in Figure \ref{fig:practical_vectors}.
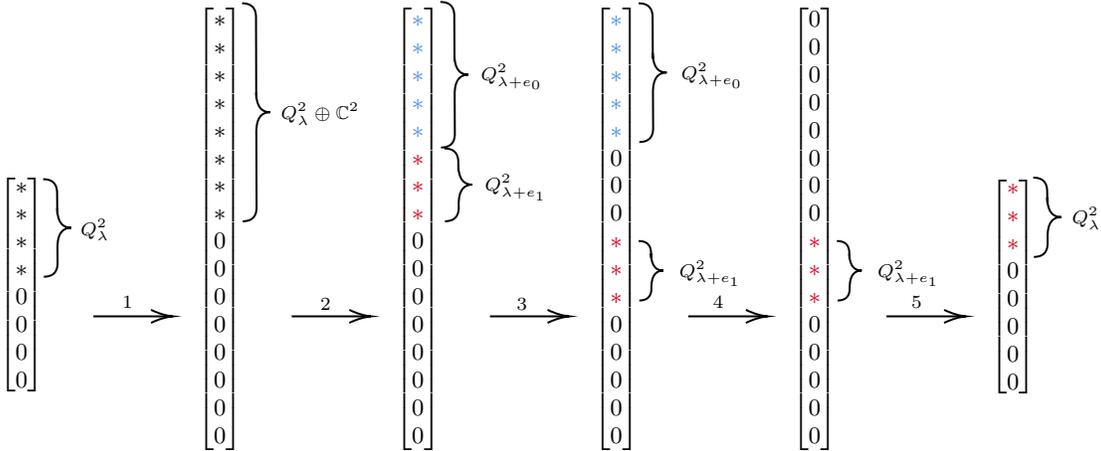
\begin{figure}[h]
    \centering
    \tikzset{every picture/.style={line width=0.75pt}} 

\begin{tikzpicture}[x=0.75pt,y=0.75pt,yscale=-1,xscale=1]

\draw   (35,140) .. controls (39.67,139.97) and (41.98,137.62) .. (41.95,132.95) -- (41.9,125.34) .. controls (41.85,118.67) and (44.16,115.32) .. (48.83,115.29) .. controls (44.16,115.32) and (41.81,112.01) .. (41.76,105.34)(41.78,108.34) -- (41.71,97.73) .. controls (41.68,93.06) and (39.33,90.75) .. (34.66,90.78) ;
\draw    (60,160) -- (98,160) ;
\draw [shift={(100,160)}, rotate = 180] [color={rgb, 255:red, 0; green, 0; blue, 0 }  ][line width=0.75]    (10.93,-3.29) .. controls (6.95,-1.4) and (3.31,-0.3) .. (0,0) .. controls (3.31,0.3) and (6.95,1.4) .. (10.93,3.29)   ;
\draw   (135.33,112) .. controls (140,111.99) and (142.32,109.65) .. (142.31,104.98) -- (142.19,66.98) .. controls (142.17,60.31) and (144.49,56.97) .. (149.16,56.96) .. controls (144.49,56.97) and (142.15,53.65) .. (142.13,46.98)(142.14,49.98) -- (142.02,8.98) .. controls (142.01,4.31) and (139.67,1.99) .. (135,2) ;
\draw    (160,160) -- (198,160) ;
\draw [shift={(200,160)}, rotate = 180] [color={rgb, 255:red, 0; green, 0; blue, 0 }  ][line width=0.75]    (10.93,-3.29) .. controls (6.95,-1.4) and (3.31,-0.3) .. (0,0) .. controls (3.31,0.3) and (6.95,1.4) .. (10.93,3.29)   ;
\draw   (235.33,74.78) .. controls (240,74.76) and (242.32,72.42) .. (242.3,67.75) -- (242.21,48.03) .. controls (242.18,41.36) and (244.5,38.02) .. (249.17,38) .. controls (244.5,38.02) and (242.15,34.7) .. (242.12,28.03)(242.13,31.03) -- (242.03,8.3) .. controls (242.01,3.63) and (239.67,1.31) .. (235,1.33) ;
\draw   (237.33,112) .. controls (242,111.96) and (244.31,109.61) .. (244.27,104.94) -- (244.25,103.57) .. controls (244.19,96.9) and (246.49,93.55) .. (251.16,93.51) .. controls (246.49,93.55) and (244.13,90.24) .. (244.07,83.57)(244.1,86.57) -- (244.06,82.21) .. controls (244.02,77.54) and (241.67,75.23) .. (237,75.27) ;
\draw    (260,160) -- (298,160) ;
\draw [shift={(300,160)}, rotate = 180] [color={rgb, 255:red, 0; green, 0; blue, 0 }  ][line width=0.75]    (10.93,-3.29) .. controls (6.95,-1.4) and (3.31,-0.3) .. (0,0) .. controls (3.31,0.3) and (6.95,1.4) .. (10.93,3.29)   ;
\draw   (335.34,72) .. controls (340.01,71.98) and (342.33,69.64) .. (342.31,64.97) -- (342.22,46.97) .. controls (342.19,40.3) and (344.5,36.96) .. (349.17,36.93) .. controls (344.5,36.96) and (342.16,33.64) .. (342.13,26.97)(342.14,29.97) -- (342.04,8.97) .. controls (342.02,4.3) and (339.68,1.98) .. (335.01,2) ;
\draw   (335.34,152) .. controls (339.45,151.95) and (341.49,149.87) .. (341.44,145.75) -- (341.44,145.75) .. controls (341.37,139.87) and (343.4,136.91) .. (347.52,136.86) .. controls (343.4,136.91) and (341.31,133.99) .. (341.25,128.11)(341.28,130.75) -- (341.25,128.11) .. controls (341.2,123.99) and (339.12,121.95) .. (335,122) ;
\draw    (360,160) -- (398,160) ;
\draw [shift={(400,160)}, rotate = 180] [color={rgb, 255:red, 0; green, 0; blue, 0 }  ][line width=0.75]    (10.93,-3.29) .. controls (6.95,-1.4) and (3.31,-0.3) .. (0,0) .. controls (3.31,0.3) and (6.95,1.4) .. (10.93,3.29)   ;
\draw   (435.34,152) .. controls (439.45,151.95) and (441.49,149.87) .. (441.44,145.76) -- (441.44,145.76) .. controls (441.37,139.87) and (443.4,136.91) .. (447.52,136.86) .. controls (443.4,136.91) and (441.31,133.99) .. (441.25,128.11)(441.28,130.76) -- (441.25,128.11) .. controls (441.2,123.99) and (439.12,121.95) .. (435,122) ;
\draw    (460,160) -- (498,160) ;
\draw [shift={(500,160)}, rotate = 180] [color={rgb, 255:red, 0; green, 0; blue, 0 }  ][line width=0.75]    (10.93,-3.29) .. controls (6.95,-1.4) and (3.31,-0.3) .. (0,0) .. controls (3.31,0.3) and (6.95,1.4) .. (10.93,3.29)   ;
\draw   (534.33,130) .. controls (539,129.96) and (541.31,127.61) .. (541.27,122.94) -- (541.25,119.94) .. controls (541.19,113.27) and (543.49,109.92) .. (548.16,109.88) .. controls (543.49,109.92) and (541.13,106.61) .. (541.08,99.94)(541.11,102.94) -- (541.05,96.94) .. controls (541.02,92.27) and (538.67,89.96) .. (534,90) ;

\draw (11,86.4) node [anchor=north west][inner sep=0.75pt]  [font=\small]  {$\begin{bmatrix}
*\\
*\\
*\\
*\\
0\\
0\\
0\\
0
\end{bmatrix}$};
\draw (52,109.21) node [anchor=north west][inner sep=0.75pt]  [font=\scriptsize]  {$Q_{\lambda }^{2}$};
\draw (73,148.4) node [anchor=north west][inner sep=0.75pt]  [font=\scriptsize]  {$1$};
\draw (111,2.4) node [anchor=north west][inner sep=0.75pt]  [font=\small]  {$\begin{bmatrix}
*\\
*\\
*\\
*\\
*\\
*\\
*\\
*\\
0\\
0\\
0\\
0\\
0\\
0\\
0\\
0
\end{bmatrix}$};
\draw (152.99,50.68) node [anchor=north west][inner sep=0.75pt]  [font=\scriptsize]  {$Q_{\lambda }^{2} \oplus \mathbb{C}^{2}$};
\draw (173,149.4) node [anchor=north west][inner sep=0.75pt]  [font=\scriptsize]  {$2$};
\draw (211,2.4) node [anchor=north west][inner sep=0.75pt]  [font=\small]  {$\begin{bmatrix}
\textcolor[rgb]{0.29,0.56,0.89}{*}\\
\textcolor[rgb]{0.29,0.56,0.89}{*}\\
\textcolor[rgb]{0.29,0.56,0.89}{*}\\
\textcolor[rgb]{0.29,0.56,0.89}{*}\\
\textcolor[rgb]{0.29,0.56,0.89}{*}\\
\textcolor[rgb]{0.82,0.01,0.11}{*}\\
\textcolor[rgb]{0.82,0.01,0.11}{*}\\
\textcolor[rgb]{0.82,0.01,0.11}{*}\\
0\\
0\\
0\\
0\\
0\\
0\\
0\\
0
\end{bmatrix}$};
\draw (254,31.08) node [anchor=north west][inner sep=0.75pt]  [font=\scriptsize]  {$Q_{\lambda +e_{0}}^{2}$};
\draw (256,86.9) node [anchor=north west][inner sep=0.75pt]  [font=\scriptsize]  {$Q_{\lambda +e_{1}}^{2}$};
\draw (272,149.4) node [anchor=north west][inner sep=0.75pt]  [font=\scriptsize]  {$3$};
\draw (311,2.4) node [anchor=north west][inner sep=0.75pt]  [font=\small]  {$\begin{bmatrix}
\textcolor[rgb]{0.29,0.56,0.89}{*}\\
\textcolor[rgb]{0.29,0.56,0.89}{*}\\
\textcolor[rgb]{0.29,0.56,0.89}{*}\\
\textcolor[rgb]{0.29,0.56,0.89}{*}\\
\textcolor[rgb]{0.29,0.56,0.89}{*}\\
0\\
0\\
0\\
\textcolor[rgb]{0.82,0.01,0.11}{*}\\
\textcolor[rgb]{0.82,0.01,0.11}{*}\\
\textcolor[rgb]{0.82,0.01,0.11}{*}\\
0\\
0\\
0\\
0\\
0
\end{bmatrix}$};
\draw (354,130.2) node [anchor=north west][inner sep=0.75pt]  [font=\scriptsize]  {$Q_{\lambda +e_{1}}^{2}$};
\draw (355,30.25) node [anchor=north west][inner sep=0.75pt]  [font=\scriptsize]  {$Q_{\lambda +e_{0}}^{2}$};
\draw (411,2.4) node [anchor=north west][inner sep=0.75pt]  [font=\small]  {$\begin{bmatrix}
0\\
0\\
0\\
0\\
0\\
0\\
0\\
0\\
\textcolor[rgb]{0.82,0.01,0.11}{*}\\
\textcolor[rgb]{0.82,0.01,0.11}{*}\\
\textcolor[rgb]{0.82,0.01,0.11}{*}\\
0\\
0\\
0\\
0\\
0
\end{bmatrix}$};
\draw (454,130.2) node [anchor=north west][inner sep=0.75pt]  [font=\scriptsize]  {$Q_{\lambda +e_{1}}^{2}$};
\draw (371,148.4) node [anchor=north west][inner sep=0.75pt]  [font=\scriptsize]  {$4$};
\draw (511,87.4) node [anchor=north west][inner sep=0.75pt]  [font=\small]  {$\begin{bmatrix}
\textcolor[rgb]{0.82,0.01,0.11}{*}\\
\textcolor[rgb]{0.82,0.01,0.11}{*}\\
\textcolor[rgb]{0.82,0.01,0.11}{*}\\
0\\
0\\
0\\
0\\
0
\end{bmatrix}$};
\draw (552,103.4) node [anchor=north west][inner sep=0.75pt]  [font=\scriptsize]  {$Q_{\lambda }^{2}$};
\draw (471,148.4) node [anchor=north west][inner sep=0.75pt]  [font=\scriptsize]  {$5$};

\end{tikzpicture}
    \caption{At the beginning of this sample iteration for $k=3$, the quantum device contains a description of irrep $\Qlt$ with $\lambda=(3,0)$ (which is $4$ dimensional) in $\lceil\log_2(3+2)\rceil$ qubits. 
    Step one sees the addition of the new qubit.
    Step 2 corresponds to the application of $U^\textrm{CG}_{(3)}:\mathcal{Q}^2_{(3)}\otimes\mathbb{C}^d \rightarrow \mathcal{Q}^2_{(4)} \oplus \mathcal{Q}^2_{(3,1)}$ ($5$ and $3$ dimensional respectively). 
    Step 3 applies the rearranging matrix so that $\ket{j}_L$ corresponds to $\mathcal{Q}^2_{(3)+\mathbf{e}_j}$.
    Step 4 is the measurement on the $L$ register, which in this case determines that $j=1$. 
    Step 5 checks Eq. \eqref{QubitRemoval} and removes a qubit accordingly.
    If Step 5 had instead determined that no qubit needs to be removed, Step 6 would be applied to bring the description of $\mathcal{Q}^2_{(3,1)}$ to the `top' of the vector between steps 4 and 5 (i.e. bringing the red entries in the bottom half, to the top half).
    This corresponds to updating the $L$ register $\ket{0}_L \leftarrow \ket{1}_L$ for the next iteration, which begins with $\lambda = (3,1)$.}
    \label{fig:practical_vectors}
\end{figure}

\begin{algorithm}[H]\label{alg:pracstreamalgo}
\SetAlgoLined
\caption{A practical algorithm for weak Schur sampling}
\SetKwInOut{Input}{Input}
\SetKwInOut{Output}{Output}
\SetKwInOut{Init}{Initialization}
\SetKwRepeat{Repeat}{repeat}{until}
\Input{A sequence of $n$ qubits}
\Output{A partition $\lambda\vdash n$}
initialization\;
receive first qubit in $\lambda=(1)$\;
initialize binary tree \textbf{PartTree}\;
\textbf{PartTree}.root$=(1)$\;
\For{$k=1$ to $n-1$}{
begin with a register of $\lceil\log_2(k+2)\rceil$ qubits encoding $\Qlt$\;
\nl receive the $(k+1)$-th qubit\;
\nl apply $\CGl$ \;
\nl apply rearranging matrix to update $L$ register: $\ket{0}_L$ and $\ket{1}_L$ \;
\nl apply measurement on the $L$ register to determine $j$ in $\lambda+\mathbf{e}_j$;

\eIf{$j=0$}{
    \textbf{PartTree}$(\lambda)$.right = $\lambda+\mathbf{e}_0$\;
    $\lambda \leftarrow \lambda+\mathbf{e}_0$\;}{
    \textbf{PartTree}$(\lambda)$.left = $\lambda+\mathbf{e}_1$\;
    $\lambda \leftarrow \lambda+\mathbf{e}_1$\;}

\nl \eIf{$\lceil\log_2(2k+4)\rceil\neq \lceil\log_2(k+3)\rceil$}{
   remove the qubit from the $L$ register\;}{
   keep all qubits\;
   \nl\If{$j=1$}{
   apply rearranging matrix to update $L$ register: $\ket{0}_L \leftarrow \ket{1}_L$
   }
  }
}
\Return $\lambda$, \textbf{PartTree}, final $\lceil\log_2(n+2)\rceil$ registers
\end{algorithm}

In the remainder of this section, we show 
\begin{theorem}\thlabel{QubitRuntime}
The protocol on $n$ qubits given by Algorithm \ref{alg:pracstreamalgo} decomposes to accuracy $\epsilon$ into $O(n^3\log_2(n/\epsilon))$ gates from the Clifford+T gate set. 
\end{theorem}
We will need the following result from \cite[Chapter\ 4.5]{NiCh10} which asserts that the error induced by a sequence of approximations scales at most linearly.
\begin{lemma}\thlabel{traceapprox}
Let $U_1,...,U_M$ be a sequence of unitary operators, and let $V_1,...,V_M$ be another such sequence such that each $V_j$ approximates $U_j$ to accuracy $\delta$ in the trace norm; that is:
\begin{align}
    ||U_j-V_j||_1:=\Tr\sqrt{(U_j-V_j)(U_j-V_j)^*}=\delta.
\end{align}
Then the product $V_1V_2\cdots V_M$ approximates the product $U_1U_2\cdots U_M$ to accuracy $\epsilon\leq M\delta$.
\end{lemma}

\begin{proof}
\textit{(of \thref{QubitRuntime})}
In \cite[Chapter\ 4.5]{NiCh10} it is shown how any two-level unitary\footnote{A unitary operator that acts non-trivially in at most two dimensions} on $n$ qubits may be decomposed into $O(n)$ $CNOT$ gates and $O(n)$ single qubit gates. 
In \cite{Selinger14}, Selinger shows that each single qubit gate may be decomposed to accuracy $\delta$ into $O(\log_2(1/\delta))$ Clifford+T gates, an improvement over the $O(\log_2^p(1/\delta))$ for $p\approx 4$ required by the Solovay-Kitaev theorem \cite[Appendix\ 3]{NiCh10}.

Ref. \cite[Chapter 4.5]{NiCh10} shows that an arbitrary unitary matrix $U$ may be decomposed into a sequence of $m$ \emph{Givens rotations} (a special case of a two-level unitary), where $m$ is the number of non-zero entries below the main diagonal of $U$. Kirby and Strauch argue in \cite{KiSt18} that each Clebsch-Gordan transform has at most two non-zero entries per row, thus $m\leq 2$. Since the Clebsch-Gordan transform $\CGl$ is a square matrix of size $2\cdot\dim\Qlt$, it follows that over the $n-1$ iterations we require a total of
\begin{align}
    \sum_{k=1}^{n-1}m\cdot(2\cdot\dim\mathcal{Q}^2_{\lambda^k}) 
    \leq \sum_{k=1}^{n-1} 4(k-2\lambda^k_1+1) 
    \leq \sum_{k=1}^{n-1} 4(k+1)
    = 2n^2+2n-4 
    = O(n^2)
\end{align}
two-level unitaries. In the summation above, $\lambda^k=(\lambda^k_0,\lambda^k_1)$ denotes the partition at the beginning of the $k$-th iteration.
Note that the permutation matrices in steps $3$ and $6$ are precisely identities of size $2\cdot\dim\Qlt$ except with permuted rows, so over the $n-1$ iterations they only contribute an additional $O(n^2)$ two-level unitaries which leaves this result unchanged.

Lastly, \thref{traceapprox} implies that if each two-level unitary in the sequence of $O(n^2)$ two-level unitaries is approximated to accuracy $\delta\leq \frac{\epsilon}{cn^2}$ for some constant $c$, then the whole sequence may be approximated to accuracy $\epsilon$. 
Hence, this requires $O(n\log_2(1/\delta))=O(n\log_2(n/\epsilon))$ Clifford+T gates, which combined with the sequence of $O(n^2)$ two-level unitaries required for all the Clebsch-Gordan transforms gives the desired result.

\end{proof}

\subsection{Weak Schur sampling algorithm for qudits}
Algorithm \ref{alg:pracstreamalgo} and the analysis for qubits in the preceding section readily generalises to a weak Schur sampling algorithm for qudits.

The registers $L$ and $Q$ used to store the partitions $\lambda+\mathbf{e}_j$ for $j\in\{0,...,d-1\}$ and the description of the respective $\mathcal{Q}^d_{\lambda+\mathbf{e}_j}$ are adjusted as follows.
\begin{itemize}
    \item The $L$ register needs to be able to describe at most $d$ partitions and therefore requires one qudit.
    \item The $Q$ register needs to contain the largest dimensional irrep that might appear during iteration $k$ (after $k+1$ qudits have been received). For qudits, the dimension of $\mathcal{Q}^d_\mu$ for $\mu\vdash k$ is given in {Stanley's hook-content formula}
    \begin{align}
        \dim\mathcal{Q}^d_\mu= \frac{\prod_{0\leq i<j\leq d-1}(\mu_i-\mu_j+i-j)}{\prod_{m=1}^{d} m!}.
    \end{align}
    This is largest when $\mu_0=k$ and $\mu_j=0$ for all $j\geq 1$, in which case
    \begin{align}
        \dim\Ql^d_{(k)}\leq (k+1)^{d-1}.
    \end{align}
    Therefore, register $Q$ requires $\lceil\log_d((k+1)+1)^{d+1}\rceil$ qudits during the $k$-th iteration.
\end{itemize}
Hence the overall memory is given by the requirements of the $(n-1)$-th iteration, $\lceil\log_d((n-1)+2)^{d+1}\rceil=O(d\log_d(n))$.

Now, we use the Solovay-Kitaev theorem for arbitrary universal sets on $d$-level systems to decompose a two-level unitary into $O(n\log_2^p(\frac{1}{\delta}))$ gates to accuracy $\delta$, where $p\approx 4$. 
Using these universal sets, our weak Schur sampling algorithm for $n$ qudits decomposes into $O(Mn\log_2(\frac{1}{\delta}))$ gates, where $M$ is the number of two-level unitaries needed to decompose the Clebsch-Gordan transforms through all the $n-1$ iterations. 
Since each $\CGl$ on $\Qld$ for $\lambda\vdash k$ is of size $d\cdot\dim\Qld\leq d(k+1)^{d-1}$, it may be decomposed into an upper bound of $(d\cdot\dim\Qld)^2 \leq d^2(k+1)^{2\cdot(d-1)}$ two-level unitaries. 
This implies
\begin{align}
    M \leq \sum_{k=1}^{n-1}d^2\cdot (k+1)^{2d-2}\leq\int_1^{n} d^2\cdot(k+1)^{2d-2} dk = O(d\cdot n^{2d-1}).
\end{align}
Using \thref{traceapprox} again, the sequence of $M= O(d\cdot n^{2d-1})$ two-level unitaries is approximated to accuracy $\epsilon$ if each two-level unitary in the sequence is approximated to accuracy $\delta\leq\frac{\epsilon}{cdn^{2d-1}}$.
Putting everything together brings the total number of gates required to implement our weak Schur sampling algorithm for qudits to
\begin{align}
    O\Big(dn^{2d-1}\cdot n\log^p_2\Big(\frac{1}{\delta}\Big)\Big) = O\Big(dn^{2d}\log_2^p\Big(\frac{n^{2d}}{\epsilon}\Big)\Big).
\end{align}

\section{Conclusion and discussion}
We provide an algorithm for weak Schur sampling which takes a stream of $n$ qudits as input and performs weak Schur sampling, returning the Young label $\lambda$ and the label $p_\lambda \in \Pl$ best describing the system of $n$ qudits. 
This algorithm implements weak Schur sampling to accuracy $\epsilon$ using $O(n^3\log_2(\frac{n}{\epsilon}))$ qubit, respectively $O(dn^{2d}\log_2^p(\frac{n^{2d}}{\epsilon}))$ qudit for $p\approx 4$, primitive gates and requires $O(\log_d n)$ qudits of memory.
Therefore, our weak Schur sampling algorithm incurs an advantage in the number of operations and an exponential advantage in memory over previous weak Schur sampling algorithms including generalised phase estimation. 
Moreover, the algorithm is streaming so it may be implemented even if only a partial amount of the input is available to the quantum machine.

We expect our algorithm to find uses beyond the NISQ era, in error corrected devices with limited quantum memory; where the logarithmic memory requirement can be taken advantage of to perform weak Schur sampling on systems with a number of qubits exponential to that of the device's memory. 

The memory savings also suggest an efficient classical simulation algorithm for weak Schur sampling:
Indeed, the hurdle when trying to simulate a quantum system of $O(n)$ qubits in a classical computer is that the required parameters scale as $O(2^n)$. 
However, logarithmic memory size requires only $2^{O(\log n)}=O(n)$ parameters. 
In fact, \cite{HS18} already provides an efficient classical simulation algorithm for the class of \textit{quantum Schur sampling circuits}, which we accomplish as a corollary of our logarithmic memory implementation.

{\bf Acknowledgements.}
E.C.M. was funded by the National Research Foundation, Singapore and A*STAR under its CQT Bridging Grant.
L.M. was funded by the European Union under the Grant Agreement No 101078107, QInteract and VILLUM FONDEN via the QMATH Centre of Excellence (Grant No 10059) and Villum Young Investigator grant (No 37532).

\newpage

\bibliographystyle{quantum}
\bibliography{refs}

\end{document}